\documentclass[runningheads, envcountsame, a4paper]{llncs}

\usepackage{times}
\usepackage{soul}
\usepackage{url}
\usepackage[utf8]{inputenc}
\usepackage[small]{caption}
\usepackage{graphicx}
\usepackage{booktabs}
\usepackage{algorithm}
\usepackage{algorithmic}
\usepackage{thmtools}
\usepackage{thm-restate}

\usepackage{graphicx}
\usepackage{bussproofs}
\usepackage{amssymb}
\usepackage{amsmath}
\usepackage{amsfonts}
\usepackage{stmaryrd}
\title{From Actions to Obligations: \\ A Deontic Action Model Logic}

\author{Giorgio Cignarale\orcidID{0000-0002-6779-4023}}
\authorrunning{G. Cignarale}

\institute{TU Wien, Vienna, Austria\\
\email{giorgio@logic.at}}

\usepackage{xspace}
\usepackage{tikz-cd}
\usepackage{tikz}
\usepackage{pgf}
\usetikzlibrary{arrows, automata}
\usetikzlibrary{calc,trees,positioning,arrows,chains,shapes.geometric, decorations.pathreplacing,decorations.pathmorphing,shapes,matrix,shapes.symbols,shapes,snakes,automata,backgrounds,petri,arrows.meta,positioning,arrows,calc,matrix,shapes.multipart,decorations.text,calc,arrows.meta}
\tikzstyle{nod}= [circle, draw,inner sep=0pt, minimum size=0.5cm]
\tikzstyle{m}=[circle, thin, draw, minimum size=14mm,inner sep=3pt]
\tikzset{reflexive left/.style={->,loop,looseness=10,in=160,out=220},
reflexive right/.style={->,loop,looseness=10,in=20,out=320}}

\usepackage{hyperref}
\usepackage{multicol}
\usepackage[nameinlink]{cleveref}

\newcommand{\La}{\mathcal{L}}
\newcommand{\calM}{\mathcal{M}}

\newcommand{\E}{\mathbb{E}}

\newcommand{\SP}{O}

\newcommand{\actions}{\mathcal{A}}
\renewcommand{\models}{\vDash}

\usepackage{amsmath,amssymb}

\makeatletter
\newcommand{\genericdot@}[2]{
  \mathbin{\mathpalette\genericdot@@{{#1}{#2}{\cdot}}}%
}
\newcommand{\genericdotop@}[2]{
  \DOTSB\mathop{\mathpalette\genericdot@@{{#1}{#2}{\boldsymbol{\cdot}}}}\slimits@
}
\newcommand{\genericdot@@}[2]{\genericdot@@@#1#2}
\newcommand{\genericdot@@@}[4]{%
  \vphantom{#3}%
  \begingroup
  \sbox\z@{$\m@th#1#3$}%
  \ooalign{%
    \usebox{\z@}\cr
    \hidewidth
    \raisebox{#2\ht\z@}[\z@][\z@]{$\m@th#1#4$}%
    \hidewidth\cr
  }%
  \endgroup
}
\newcommand{\cupdot}{\genericdot@{0.3}{\cup}}
\newcommand{\veedot}{\genericdot@{0.3}{\vee}}
\newcommand{\capdot}{\genericdot@{-0.2}{\cap}}
\newcommand{\wedgedot}{\genericdot@{-0.2}{\wedge}}

\newcommand{\bigcupdot}{\genericdotop@{0.25}{\bigcup}}
\newcommand{\bigveedot}{\genericdotop@{0.25}{\bigvee}}
\makeatletter

\begin{document}

\maketitle              

\begin{abstract}

We introduce the Deontic Action Model Logic (\textbf{DAML}), a dynamic modal framework for reasoning about obligations over actions in multi-agent systems. \textbf{DAML} extends the epistemic Action Model Logic by incorporating deontic evaluation mechanisms that assess agents’ actions in terms of 
both the desirability and the likelihood of their outcomes.
Obligations arise for those actions that maximize expected deontic value among an agent’s available alternatives at a given decision point, yielding a formal account
for reasoning about conditional and context-sensitive obligations in settings involving strategic interaction and incomplete information.
\textbf{DAML} supports principled action selection in norm-governed multi-agent systems, and is the first such framework to derive these obligations using the action model logic machinery.
We provide an axiomatization of the logic and prove soundness and completeness with respect to its semantics. Finally, we demonstrate the expressive power of our framework through applications to the Miners’ Puzzle and other multi-agent deontic scenarios.

\end{abstract}


\section{Introduction}
\label{sec:Intro}

Deontic logic \cite{Gabbay2013HandbookOD,sep-logic-deontic} studies normative concepts such as obligations, permissions, and prohibitions, and has been widely applied in philosophy, legal theory, and computer science, in particular for the formal specification of normative and multi-agent systems. Formally, it extends propositional logic with modal operators expressing normative force, such as obligation ($O\varphi$), permission ($P\varphi$), and prohibition ($F\varphi$), often further enriched with conditional modalities $O(\varphi\mid\psi)$ to capture context-dependent obligations.

There are two extensions of deontic logic that are relevant to our study.
First, 
Seeing To It That (STIT) logic \cite{Horty2001,sep-logic-action} is a family of modal logics designed to formally represent agency, choice, and responsibility. Its central idea is to capture what it means for an agent to see to it that certain state of affairs holds. 
In particular, the deontic STIT framework is rooted in a utilitarian approach to choice-making \cite{Horty2001} and is employed in multi-agent settings \cite{ConflictingSTIT}.
Second, dynamic deontic logics \cite{Meyer1987-MEYADA,Segerberg2012-DAL} introduce dynamic operators to evaluate permissions and obligations in relation to agents’ actions. Crucially, however, these frameworks do not explicitly incorporate epistemic considerations.

Epistemic logic (EL) \cite{Hintikka1962-HINKAB,van2015handbook} provides a systematic method for representing and reasoning about epistemic states, i.e., what agents know, believe, or consider possible, and has been wildly successful in a number of different areas, such as computer science, economics, philosophy and game theory.
Dynamic epistemic logic (DEL) \cite{plaza,ditmarsch2007dynamic}  extends EL with epistemic updates, such as public announcements, private observations and other complex epistemic actions.
Perhaps the most well-known of these logics is public announcement logic (PAL) formulated by Plaza \cite{plaza},  wherein
an announcement of a true proposition 
$\varphi$ eliminates all possible worlds incompatible with it, thereby refining the epistemic model.
Subsequent generalizations \cite{Baltag1998}, led to the formulation of action model logic (AML), which captures a broader class of epistemic actions beyond public announcements, including private and semi-private information updates. 

With a few notable exceptions \cite{lomuscio2003deontic,Pacuit2006-PACTLO-4}, systematic interactions between epistemic and deontic frameworks remain limited. Nevertheless, many deontic problems crucially rely on an underlying epistemic framework or, at the very least, involve some degree of uncertainty. A well-known example is the Miners’ Puzzle:

\begin{example}[Miners' Puzzle \cite{Kolodny2010-KOLIAO}]
\label{ex:miners}
Ten miners are trapped either in shaft A or in shaft B, but we do not know which.
Flood waters threaten to flood the shafts. We have enough sandbags to block one shaft, but not both. If we block one shaft, all the water will go into the  other shaft, killing any miners inside it.
If we block neither shaft, both shafts will fill halfway with water, and just one miner, the lowest in the shaft,
will be killed.
\end{example}

The goal of this paper is to provide a combined framework, where the epistemic effects of agent's actions are evaluated, providing, in turn, obligations towards the actions that maximize an expectation function.

An argument for a “Bayesian” semantics of deontic modalities has been proposed by Lassiter \cite{lassiter2016linguistic,Lassiter2017}.
The core idea is that deontic modals have a semantical structure that is scalable, and that deontic scales are formally identical to the expected utility scales used in Bayesian decision theory, with the difference that the function assigns an \emph{expected moral value} instead of a personal utility.
The expectation function is computed by (i) assigning a value $V(w)$ to each possible world $w\in W$ in a model, representing how morally or practically desirable it would be if all of the facts of the world were arranged as in $w$; and (ii) by using a probability measure $P$, representing how likely a given possible world is. From (i) and (ii), \cite{lassiter2016linguistic} argues that the expected moral value of a proposition in a set of worlds can be evaluated by the the average value of the worlds where the proposition is true:

$$\E(\varphi) = \sum_{w \in \llbracket\varphi\rrbracket} V(w) \times P(\{w\}|\varphi) $$

Where $\llbracket\varphi\rrbracket$ is the truth set of that formula, i.e., the set of worlds where $\varphi$ holds.

We endorse this idea and we provide a possible-world semantics accounting for the expected moral value framework proposed by \cite{lassiter2016linguistic}.\\

\textbf{Main contributions:} 
\begin{itemize}
 
   \item Inspired by Lassiter \cite{lassiter2016linguistic}, we model deontic statements in a scalar way, by assigning
   to each world in a Kripke model a desirability value that we use to compute an expected deontic value of a model (or a portion thereof).

   \item Following the dynamic deontic logic's approach \cite{Meyer1987-MEYADA,Segerberg2012-DAL}, we treat actions as first class citizens of deontic statements and we employ the Action Model Logic (AML) machinery \cite{Baltag1998} to formally reason about the epistemic consequences of agents’ actions in order to compare their resulting outcomes. More concretely, the kind of questions that we aim at answering are: 
   "If agent $i$ were to perform action $\alpha_i$, would the resulting outcome be more desirable than that obtained by performing action $\beta_i$?"

   \item 
   We introduce a new binary deontic modality $\SP_i(\alpha_i|\varphi)$ that reads as "agent $i$ ought to perform action $\alpha_i$ under condition $\varphi$". Its truthfulness depends both on (i) $\varphi$ holding after action $\alpha_i$, and (ii) action $\alpha_i$ leading to the outcome with the highest expected deontic value.
   In particular, $\SP_i(\alpha_i|\varphi)$ adheres to the STIT principle \cite{Horty2001}: agent $i$ sees to it that the best expected outcome is obtained, meaning that she ought to perform action $\alpha_i$ under the assumption that $\varphi$.

   \item We prove the soundness and completeness of the resulting proof system of the logic and we illustrate its usefulness with some examples.

\end{itemize}

The resulting Deontic Action Model Logic (\textbf{DAML}) provides a formal reasoning framework for agents that must select actions under epistemic uncertainty while taking normative constraints into account. By explicitly representing and comparing the expected deontic consequences of alternative actions, the logic supports principled decision-making in norm-governed multi-agent environments.
In particular, an action is considered most desirable when it both (i) significantly reduces uncertainty and (ii) leads to normatively preferable consequences. Crucially, the framework also supports reasoning in intermediate cases where neither criterion is fully optimized, that is, under incomplete information.

\textbf{Paper organization}
The rest of this paper is organized as follows: \Cref{sec:formalism} introduces some preliminary definitions for our framework. The novel Deontic Action Model Logic \textbf{DAML} is presented in \Cref{sec:DAML}, and its soundness and completeness are shown in \Cref{sec: reduction}.
\Cref{sec:examples} is dedicated to examples, including the famous Miner's puzzle (\Cref{sec:miner}), and a more complex multi-agent scenario (\Cref{sec:med}). 
Finally, some conclusions are offered in \Cref{sec:conclusion}.

\section{Expectation functions and submodels}
\label{sec:formalism}

In this section, we elaborate some preliminary concepts and definitions of our novel \textbf{DAML}. 
We assume a finite  non-empty set of agents $\Pi= \{i,j,\ldots \}$ and a fixed non-empty set of possible actions $\alpha_i \in \actions$, indexed by agents in $\Pi$ indicating the agent that can perform said action.
While it is possible to create more complex actions using PDL \cite{sep-logic-dynamic}, we only make use of action concatenation $\alpha_i;\beta_i$, here representing the execution of action $\alpha_i$ followed by $\beta_i$.
As standard in EL, we rely on Kripke models $\calM = \langle W, R_i, V \rangle$, where $W \ne \varnothing$ is a non-empty set of \emph{possible worlds}, called \emph{domain}, $R_i \subseteq 2^{W\times W}$ is a binary \emph{accessibility relation}, one for each agent $i \in \Pi$ and $V : Prop \to 2^W$ is a \emph{valuation function} that assigns to each atomic proposition $p \in Prop$ a set of worlds $V(p) \subseteq W$ where $p$ holds.
A \emph{pointed graded Kripke model} is a pair $(\calM, v)$ where $\calM$ is a Kripke model with $v \in W$. We use $R_i(v)$ to denote the set of all worlds that are $i$-accessible from $v$.

Unless specified otherwise, this paper relies on \textbf{S5} relations to model knowledge, i.e., where the accessibility relations are reflexive ($uR_iu$), transitive (if $uR_i v$ and $vR_i w$ then $uR_i w$) and euclidean (if $uR_i v$ and $uR_i w$ then $vR_i w$).

We continue with some introductory definitions, closely following \cite{ditmarsch2007dynamic}, with the distinction that we use agent-indexed actions $\alpha_i$.

\begin{definition}[Pointed action model]
\label{def:AM}
    An \emph{action model} is a triple $U = \langle E, Q_i, pre \rangle$ where
    $E \ne \varnothing$ is a finite set of \emph{action points}, each indexed by an agent, $Q_i \subseteq 2^{E\times E}$ is a binary \emph{accessibility relation}, one for each agent 
    and the \emph{precondition function} $pre: E \rightarrow \La$ assigns the precondition  $pre(\beta_i)\in \La$ that is necessary for an event $\beta_i \in E$ to happen.\footnote{ We return to the role of the precondition function in \Cref{rem:actual}.}
    A \emph{pointed action model} is a pair $(U,\alpha_i)$ with $\alpha_i \in E$.

\end{definition}

An action model simulates the results of an epistemic action in a Kripke model via the \emph{product update} operation:

\begin{definition}[Product update]
\label{def:prod_up}
    The \emph{(restricted modal) product update} of  a Kripke model $\calM = \langle W,R_i,V \rangle$ with an action model $U = \langle E, Q_i, pre \rangle$ is   a Kripke model $\calM \otimes U := \langle S',R_i',V' \rangle$ where:
\begin{itemize}
    \item $W' := \{(v, \beta_i) \in W \times E \mid \calM, v \vDash pre(\beta_i) \}$,
    \item $R_i' := \{\bigl((v, \beta_i),(u, \gamma_i)\bigr)\in W'\times W' \mid (v, u) \in R_i$ and $(\beta_i, \gamma_i) \in Q_i\}$,
    
    \item $V'(p) := \{(v, \beta_i) \in W' \mid v \in V(p)\}$.
\end{itemize}
If $W' = \varnothing$, the product update is undefined.
\end{definition}

Intuitively, the result of a product update is a Kripke model that preserves those worlds that satisfy the preconditions of the actions. Their epistemic effects can be evaluated in the semantics:
$\calM, w \vDash [U,\alpha_i]\varphi$, if{f} $\calM,w \models pre(\alpha_i)$ implies $\calM \otimes U, (w,\alpha_i) \vDash \varphi$ \cite{ditmarsch2007dynamic}. \Cref{tab:PALtheory} shows the axiomatic system of  AML. We adopted the definition in \cite{ditmarsch2007dynamic} using our notation.
In particular, the action model composition operation $U \circ U'$ represents the effects of applying consecutive action models with a single combined update model \cite{ditmarsch2007dynamic}:

\begin{definition}[Action model composition]
\label{def:compo}
    Given two action models $U =(E, Q_i, pre)$ and $U'=(E', Q', pre')$, their composition $U\circ U'=(E'', Q_i'', pre'')$ is defined as: $E''=E\times E'$, $(\alpha_i;\alpha'_i) Q_i'' (\beta_i;\beta'_i)$ iff $\alpha_i Q_i \beta_i$ and $\alpha'_i Q'_i \beta'_i$, and $pre''(\alpha_i;\alpha'_i)= \langle U,\alpha_i\rangle pre'(\alpha'_i)$.

\end{definition}

\begin{table}[]
    \centering
    \begin{tabular}{c l}
       \textbf{Taut}  &  All instantiations of propositional tautologies \\
       \textbf{S5}.  &  Axioms of \textbf{S5} modal logic \\
       \textbf{AM1}.  & $[ U,\alpha_i] p \leftrightarrow (pre(\alpha_i) \rightarrow p)$ \\
       \textbf{AM2}.  & $[ U,\alpha_i]\neg \varphi \leftrightarrow (pre(\alpha_i) \rightarrow \neg [ U,\alpha_i] \varphi)$ \\
       \textbf{AM3}.  & $[ U,\alpha_i] (\varphi \wedge \theta) \leftrightarrow ([ U,\alpha_i]\varphi \wedge [ U,\alpha_i]\theta)$ \\
       \textbf{AM4}.  &  $[ U,\alpha_i] K_i \varphi \leftrightarrow (pre(\alpha_i) \rightarrow \bigwedge_{\alpha_i Q_i \beta_i} K_i [ U,\alpha_i]\varphi)$ \\
       \textbf{AM5}.  &  $[ U,\alpha_i][ U',\beta_i] \varphi \leftrightarrow [U\circ U',\alpha_i;\beta_i] \varphi$
    \end{tabular}
    
\begin{prooftree}
\AxiomC{$\varphi\rightarrow \psi$}
\AxiomC{$\varphi$}
\RightLabel{MP}
\BinaryInfC{$\psi$}
\DisplayProof
\qquad
\AxiomC{$\varphi$}
\RightLabel{KN}
\UnaryInfC{$K_i\varphi$}
\DisplayProof
\qquad
\AxiomC{$\varphi$}
\RightLabel{$ U,\alpha$N}
\UnaryInfC{$[ U,\alpha_i]\varphi$}

\end{prooftree}
    
    \caption{Axiomatic system of AML.}
    \label{tab:PALtheory}
\end{table}

In our framework, deontic statements are evaluated by comparing the results of different actions for different agents. Thus, we need to introduce expectation functions that range over only a portion of the model, based on (i) agent's accessibility relation and (ii) models generated by actions.
The following definitions are meant to capture exactly these features, through the notion of \textit{deontic expectation function} and the various notions of submodels introduced below.

Following the ideas in \cite{lassiter2016linguistic}, we first introduce  a \textit{desirability function} $f: W \to \mathbb{N}$, which assigns natural numbers to possible worlds, representing their deontic desirability.
 We call a Kripke model equipped with a desirability function a \textit{graded Kripke model} $\calM = \langle W, R_i, V,f \rangle$.\footnote{We assume that each world  in a graded Kripke model has an "objective" desirability value, meaning that all agents have the same opinion about these values. This assumption can easily be relaxed e.g. by using weaker logics than \textbf{S5}. Exploring alternative ways to model subjective desirability values is left for future work.}
As argued by \cite{lassiter2016linguistic}, each deontic value $f(w)$ is meaningful only in relation to the deontic value of the other worlds $f(u)$.
In particular, a product update between a graded Kripke model and an action model leaves the desirability function unchanged: given $\calM =\langle W,R_i,V,f\rangle$, $\calM \otimes U = \langle W',R_i',V',f' \rangle $ is such that $f'(w,\alpha_i) = f(w)$.

Graded Kripke models allow to use an agent-based expectation function on models, that we label deontic expectation function:

\begin{definition}[Deontic expectation function]
Given an agent $i \in \Pi$ and a pointed Kripke model $(\calM,v)$, $i$'s expectation function on that model $\E_i(\calM)$ is defined as follows:

$$\E_i(\calM) := \sum_{w \in W} f(w) \times \frac{1}{|R_i(v)|} $$

\end{definition}

Intuitively, under the assumption that each possible world is equally possible\footnote{We aim at introducing a probability value to each world in future work.}, the deontic expectation function of an agent $i$ on a model $\calM$ weights an average of the value $f(w)$ for all worlds $w \in W$ that are accessible to that agent from a given world $v \in W$.\footnote{Naturally, using an expectation function is only one way to model agent's attitudes towards desirability values. Exploring other attitudes is left for future work.}
We call the value assumed by the deontic expectation function \textit{expected deontic value}.

The deontic evaluation function can also be applied to only portions of a Kripke model, for which we use the notion of submodel:\footnote{Similar concepts can be found in \cite{blackburn_rijke_venema_2001,APU}.}

\begin{definition}[Submodel]
    Let $(\calM,v)$ be a pointed model such that $R_i(v) \ne \varnothing$.
    The \emph{submodel} of $\calM$ constructed from $v$ is the Kripke model $\calM^v := \langle W', R_i', V'\rangle$ such that:
    \begin{itemize}
        \item $W'\subseteq W$ is the set of all worlds $u$ such that  $v R_{i_0} u_1 R_{i_1} u_2 R_{i_2} \dots u_{k}R_{i_k} u$  for some $k\geq 0$, with worlds $u_1,\dots,u_{k} \in W$ and agents $i_0,\dots, i_k \in \Pi$;
        \item $R_i' := R_i \cap (W' \times W')$ for each $i \in \Pi$;
        \item $V'(p) := V(p)\cap W'$ for each $p \in Prop$;
    \end{itemize}

We call \textit{agent-based submodel} a submodel constructed using accessibility relations of a single agent $i \in \Pi$ from a given point $v \in W$, labeled $\calM_i^v$.
\end{definition}

Since every submodel of a Kripke model is a Kripke model, the deontic expectation function can be formulated analogously for agent-based submodels $\calM_i^v$, written $\E_i(\calM_i^v)$.
Similar considerations can be extended to Kripke models resulting from product updates. We call the submodels of these Kripke models \textit{action-generated submodels}:

\begin{definition}[Action-generated submodel]

Given a graded Kripke model $\calM \otimes U = \langle W, R_i, V  . f\rangle$ obtained via the product update operation between a  graded pointed Kripke model $(\calM,v)$ and an action model $U= \langle E, Q_i, pre\rangle$ its \emph{action generated submodels} are the Kripke models $\calM^{v,\alpha_i} = \langle W^{\alpha_i}, R^{\alpha_i}, V^{\alpha_i}  f^{\alpha_i}\rangle$ such that:

    \begin{itemize}
        \item  $W^{\alpha_i} \subseteq W$ is the set of all worlds $(u,\alpha_i)$ such that  $(v,\alpha_i) R^{\alpha_i}_{i_0} (u_1,,\alpha_i) R^{\alpha_i}_{i_1} (u_2,\alpha_i) \\ R^{\alpha_i}_{i_2} \dots (u_{k},\alpha_i) R^{\alpha_i}_{i_k} (u,\alpha_i)$  for some $k\geq 0$, worlds $(u_1,\alpha_i),\dots,(u_{k},\alpha_i) \in W$, and agents $i_0,\dots, i_k \in \Pi$;
        \item $R_i^{\alpha_i} := R_i \cap (W^{\alpha_i} \times W^{\alpha_i})$ for each $i \in \Pi$;
        \item $V^{\alpha_i}(p) := V(p)\cap W^{\alpha_i}$ for each $p \in Prop$;
        \item  $f^{\alpha_i} = f$
    \end{itemize}
For all actions $\alpha_i\in E$.
We call \textit{agent-based action-generated submodel} an action-generated submodel constructed using accessibility relations of a single agent $i \in \Pi$ from a given point $(v,\alpha_i) \in W$, labeled $\calM_i^{v,\alpha_i}$.
\end{definition}

\begin{remark}[Action-generated submodel notation]
The notation $\calM^{w,\alpha_i}_i$ indicates the agent based action generated submodel of the Kripke model $\calM \otimes U$, which is not pointed. The superscript $w,\alpha_i$ denotes the point from which the submodel is constructed, but is not itself necessarily the point of evaluation. To make that explicit we write $\calM^{w,\alpha_i}_i, (w,\alpha_i)$.

\end{remark}

Agent-based action generated submodels are crucial in our framework to evaluate deontic statements, for which we  formulate a corresponding deontic expectation function:

\begin{definition}[Deontic expectation function for agent-based action generated submodels]
Given an agent $i \in \Pi$ and a  graded Kripke model $\calM \otimes U = \langle W, R_i, V,  f\rangle$ obtained via the product update with $U= \langle E, Q_i, pre\rangle$ such that $\alpha_i \in E$, we define the deontic expectation function for any of its agent-based action generated submodels $\calM_i^{v,\alpha_i} = \langle W^{\alpha_i}, R_i^{\alpha_i}, V^{\alpha_i}  ,f^{\alpha_i} \rangle $ as follows:

$$\E_i(\calM^{v,\alpha_i}_i) := \sum_{(w,\alpha_i) \in W^{\alpha_i}} f^{ \alpha_i}(w, \alpha_i) \times \frac{1}{|R_i^{\alpha_i}(v,\alpha_i)|}$$

\end{definition}

Before moving to the definitions of the language and the semantics of \textbf{DAML}, we state some important remarks:

\begin{remark}[Distinguishability of actions]
Considering the deontic expectation function of submodels generated by actions is particularly convenient in our framework, as we usually assume our action models to be only \textit{reflexive}, meaning that the actions $\alpha_i \in E$ are \emph{distinguishable} from each other, at least for the acting agent.
This reflects the fact that the acting agent \textit{knows} the results of the different actions that she can perform.
A single-agent example of such action model is provided in \Cref{fig:Miners_AM} (top right).
\end{remark}

\begin{remark}[Decision point]
In the following, we impose some important restrictions on our action models.
First, each action model $U=\langle E, Q_i, pre\rangle$ expresses a set of actions available to a single agent $i$ at a given point of the system evolution, that we call \textit{decision point}. In particular, each decision point $U$ consists of at least two actions $|E| \geq 2$. 
We also assume that none of the actions at a given decision point leads to an empty (sub)model, i.e., the agents maintain consistency through their actions.
As we use action models to simulate the consequences of agent's actions, we do not use pointed action models. In this sense, our action models are closer to \textit{multi-pointed action models} \cite{Hales}, wherein one considers a set of actions instead of a single point.
Considering synchronous and parallel actions of agents is outside the scope of the current work.
\end{remark}

\section{Deontic Action Model Logic}
\label{sec:DAML}

In this section, we introduce the main definitions of \textbf{DAML}.
The final ingredient for our logic is a set of \textit{expectation atoms} $\varepsilon = \{ e^{\alpha_i}_i, e^{\beta_j}_j \ldots \}$ expressing the highest expected deontic value of action-generated submodels, where  $e^{\alpha_i}_i$ stands for $i$'s expected deontic value of the submodel generated by the action $\alpha_i$, at a given decision point, i.e., $\E_i(\calM^{w,\alpha_i}_i)$.

We can now introduce the language of \textbf{DAML}:

\begin{definition}[Language]
The language $\La^{DAML}$ is defined by the following grammar: 
    \[\varphi ::=  \ p  \ | \  e^{\alpha_i}_i \ | \ \neg\varphi \ | \  (\varphi \wedge \varphi) \ | \ K_i \varphi  \ | \ \langle U,\alpha_i\rangle\varphi \ | \ \SP_i( U,\alpha_i| \varphi) \]

With $p \in Prop$, $ e^{\alpha_i}_i \in \varepsilon$ and $\alpha_i \in \mathcal{A}$.
We use standard propositional abbreviations, $\bot := p \wedge \neg p$, $\top := \neg \bot$, $\hat{K}_i\varphi := \neg K_i\neg \varphi$ and $[ U,\alpha_i] \varphi := \neg \langle U,\alpha_i\rangle\neg \varphi$. We denote with $\La$ the fragment of $\La^{DAML}$ without the new modal operator $\SP_i( U,\alpha_i | \varphi)$.

\end{definition}

The two new elements that are added to the language of AML are the expectation atoms $e^{\alpha_i}_i$ and the novel modality $\SP_i ( U,\alpha_i|\varphi)$, which reads as "agent $i$ ought to perform action $\alpha_i$ under the assumption that $\varphi$".
Their meaning is captured by the following semantics:

\begin{definition}[Semantics of $\La^{DAML}$]
\label{def:semLA}
Given a graded pointed Kripke model $(\calM,v) = (\langle W, R_i, V  ,f\rangle,v)$ with $v,w \in W$, a pointed action model $(U,\alpha_i) = (\langle E, Q_i, pre\rangle, \alpha_i)$, with $\alpha_i, \beta_i \in E$, the agent submodel $\calM^v_i$ of $(\calM,v)$ and the agent-based action generated submodel $\calM^{(v,\alpha_i)}_i$ of $\calM \otimes U, (v, \alpha_i)$ we define the following semantics:

\begin{itemize}
    \item $\calM, w \models p \text{ iff } w \in V(p)$
    \item $\calM^{w,\alpha_i}_i  ,(w,\alpha_i) \models e^{\alpha_i}_i \text{ iff } \E_i(\calM^{w,\alpha_i}_i) \geq \E_i(\calM^{w,\beta_i}_i)$ for all  $(w,\alpha_i) \in W^{\alpha_i} \ne (w,\beta_i) \in W^{\beta_i}$ \footnote{Here $W^{\alpha_i}$ and $W^{\beta_i}$ are the domains of the action generated submodels from actions $\alpha_i$ and $\beta_i$. Recall that each $U$ represents a decision point that collects only the actions of a single agent $i \in \Pi$. Additionally, using $\geq$ implies that there is always an action $\alpha_i \in E$ leading to the highest expectation value}
    \item $\calM, w \models \neg \varphi \text{ iff } \calM, w \not\models \varphi $
    \item $\calM, w \models (\varphi \wedge \psi) \text{ iff }\calM, w \models \varphi \text{ and } \calM, w \models \psi $
    \item $\calM, w \models K_i\varphi \text{ iff } \calM, w \models \varphi \text{ for all } w'\in W \text{ s.t. }
    w' \in R_i(w)$
    \item $\calM, w \models \langle U,\alpha_i \rangle \varphi$  if{f}  
    $\calM, w \models pre(\alpha_i)$ and $\calM\otimes U, (w,\alpha_i) \models \varphi$, where $\calM\otimes U := \langle W',R_i',V' \rangle$ is defined in \Cref{def:prod_up};

    \item $\calM_i^w ,w \models \SP_i ( U,\alpha_i|\varphi) $ iff $ \calM_i^w, w \models \langle  U,\alpha_i \rangle \varphi $ and $ \calM_i^{w,\alpha_i}  ,(w,\alpha_i) \models e^{\alpha_i}_i$

\end{itemize}

In particular, $\calM \otimes U$ is defined whenever $\calM \models\langle  U,\alpha_i \rangle \varphi$ is.
We say that $\calM_i^v \models \varphi$ holds if, for all $w \in W$,  $\calM_i^v,w \models \varphi$.
Given a class of models $C$, we say that $\varphi$ is valid in the class of models $C$ iff $\calM \models \varphi$ for all $\calM \in C$. \looseness=-1
\end{definition}

Intuitively, expectation atoms are true if and only if the expected deontic value for $i$ for that action-generated submodel is greater or equal the expected deontic value of the other action-generated submodels at that decision point.\footnote{Recall that it is not possible to have an expectation atom for no action, as we assumed both $\actions \ne \varnothing$ and $E \ne \varnothing$.}

The new modality $\SP_i ( U,\alpha_i|\varphi)$ is true if and only if (i) $\varphi$ is true after action $\alpha_i$, and (ii) action $\alpha_i$ leads to the action generated submodel whose expected deontic value is greater or equal the expected deontic value of the submodels generated via other actions at that decision point.\footnote{Because some actions might lead to equally good expected deontic value for the corresponding action-generated submodels, more than one ought formula might be true for an agent at a decision point. Having multiple (and often conflicting) obligations is an expected feature of deontic logics \cite{Horty2001,ConflictingSTIT} and exploring resolution mechanisms is outside the scope of this work.}
In particular, expectation atoms hold globally in a submodel, that is, they are true in every world of an agent-based action generated submodel. This is because their truth depends on expectation functions which range over submodels. Consequently, also ought formulas are true (or false) in every world, as their truth depends on the truth of the corresponding expectation atoms.
For this reason in the rest of the paper we will omit the evaluation world for these formulas, writing $\calM^{v,\alpha_i}_i  \models e^{\alpha_i}_i$ and $\calM_i^v \models \SP_i (U,\alpha_i|\varphi)$. 

The core idea is that we interpret the ought modality  with an intrinsic STIT component: agent $i$ sees to it that the highest expected outcome is obtained,  and hence ought to perform action $\alpha_i$, under the assumption that $\varphi$.

At the same time, in the truth conditions for the modality $\SP_i( U,\alpha_i | \varphi)$, $\langle U,\alpha_i \rangle \varphi$ is evaluated in the pointed model $\calM_i^v,v$ constructed from $v$.

Taken together, these features make \textbf{DAML} a framework that models agents’ hypothetical reasoning, allowing them to evaluate the consequences of alternative actions at a decision point \textit{prior} to execution and to select those actions that lead to the most desirable outcomes given their current informational state.

In the next section we prove soundness and completeness of \textbf{DAML}.

\section{Soundness and completeness}
\label{sec: reduction}

The main goal of this section is to prove that the language $\La^{DAML}$ can be reduced to the language $\La$ which does not contain deontic operators, and for which completeness is well known \cite{ditmarsch2007dynamic}.
Technically, $\La$ differs from the standard language of AML as it contains expectation atoms $e_i^\alpha$. However, this does not constitute a problem for our reduction: since the completeness of AML depends on the completeness of EL via reduction, adding a set of atoms (like $e_i^\alpha$) to EL does not alter its completeness, which in turn does not alter the completeness of AML. In addition, we avoid potential circularity by restricting the preconditions used in our action models to formulas solely from the language of AML, as per \Cref{def:AM}.

Our goal is to formulate reduction axioms showing that for every formula with a deontic operator there exist a translation of that formula without that operator, without altering its meaning. 
Before that, we introduce the following elementary lemma used in later proofs:

\begin{restatable}[]{lemma}{1}
\label{lemma:glob}

For any agent $i$ and corresponding action $\alpha_i$:

\begin{itemize}
    \item[(i)] $\calM_i^{v,\alpha_i}  \models e^{\alpha_i}_i \leftrightarrow K_i e^{\alpha_i}_i$;

    \item [(ii)]$\calM_i^v \models \SP_i( U,\alpha_i|\varphi) \leftrightarrow K_i \SP_i( U,\alpha_i|\varphi)$
\end{itemize}

\end{restatable}

\begin{proof}
$(\Rightarrow)$ (i) Since $\calM_i^{v,\alpha_i} \models e^{\alpha_i}_i$, then for all $ (w,\alpha_i) \in W$, $\calM_i^{v,\alpha_i}, (w,\alpha_i) \models e^{\alpha_i}_i$, making $K_i e^{\alpha_i}_i$ true. (ii) is proven analogously.

    $(\Leftarrow)$ (i) and (ii) follow from axiom \textbf{T} of $K_i$.
\end{proof}

We can now introduce the axiomatic system of \textbf{DAML}, 
which extends the axiom system of AML \cite{ditmarsch2007dynamic} with the following axioms for the new expectation atoms and the novel deontic modality:

\begin{definition}[Axioms of \textbf{DAML}]
\label{def:redAx}
For $i \in \Pi$, $\alpha_i, \beta_i \in \actions$, $p \in Prop$, $\varphi,\psi \in \La^{DAML}$ and $e^{\alpha_i}_i \in \varepsilon$, the axiomatic system of \textbf{DAML} is shown in \Cref{tab:DAMLtheory}. We write $DAML \vdash \varphi$ to indicate that $\varphi$ is derivable in \textbf{DAML}.

\begin{table}[]
    \centering
    \begin{tabular}{c l}
       \textbf{AML} & Axioms of AML (\Cref{tab:PALtheory}) \\
       \textbf{E1}  & $ \bigvee_{i \in \Pi}e^{\alpha_i}_i$  \\
       \textbf{E2}  & $e^{\alpha_i}_i \rightarrow K_i e^{\alpha_i}_i$  \\
       \textbf{R1}  & $\SP_i( U,\alpha_i|p) \leftrightarrow pre(\alpha_i) \wedge p \wedge e^{\alpha_i}_i$  \\
       \textbf{R2}  & $\SP_i( U,\alpha_i|\varphi \wedge \psi) \leftrightarrow \SP_i( U,\alpha_i|\varphi) \wedge \SP_i( U,\alpha_i|\psi)$ \\
       \textbf{R3}  & $\SP_i( U,\alpha_i|\neg \varphi) \leftrightarrow pre(\alpha_i) \wedge \neg \SP_i( U,\alpha_i|\varphi)$ \\
       \textbf{R4}  &  $\SP_i( U,\alpha_i|K_i\varphi) \leftrightarrow K_i \SP_i( U,\alpha_i|\varphi)$ \\
       \textbf{R5} & $\SP_i( U,\alpha_i|\langle U',\beta_i\rangle\varphi) \leftrightarrow \langle U\circ U',\alpha_i;\beta_i\rangle\varphi \wedge e^{\alpha_i}_i$ \\
       \textbf{R6}  &  $\SP_i( U,\alpha_i|\SP_i( U',\beta_i|\varphi)) \leftrightarrow \SP_i(U\circ U',\alpha_i;\beta_i|\varphi) \wedge e^{\alpha_i}_i$
    \end{tabular}
    
\begin{prooftree}
\AxiomC{$\varphi\rightarrow \psi$}
\AxiomC{$\varphi$}
\RightLabel{\textbf{MP}}
\BinaryInfC{$\psi$}
\DisplayProof
\qquad
\AxiomC{$\varphi$}
\RightLabel{\textbf{KN}}
\UnaryInfC{$K_i\varphi$}
\DisplayProof
\qquad
\AxiomC{$\varphi$}
\RightLabel{$\alpha$\textbf{N}}
\UnaryInfC{$\langle U,\alpha_i\rangle\varphi$}

\end{prooftree}
    
    \caption{Axiomatic system of \textbf{DAML}.}
    \label{tab:DAMLtheory}
\end{table}

\end{definition}

In particular, axioms \textbf{E1} and \textbf{E2} express properties of the expectation atoms. \textbf{E1} states that at least one expectation atom holds for each agent. \textbf{E2} states that if an expectation atom is true it is also known. \textbf{R1-R6} are reduction axioms.
We begin by showing soundness of the axiomatic system.

\begin{restatable}[Soundness]{theorem}{S}
\label{thm:soundness}
The axiomatic system of \Cref{def:redAx} is sound. 
\end{restatable}

\begin{proof}
Soundness for axioms of epistemic logic and of reduction axioms of AML is standard \cite{ditmarsch2007dynamic}. Here we prove soundness of the newly introduced axioms \textbf{E1-E2} and of the reduction axioms \textbf{R1-R6}:

    \begin{itemize}

   \item[\textbf{E1}] $\bigvee_{i \in \Pi}e^{\alpha_i}_i$: easily follows from the semantics of $e^{\alpha_i}_i$ (\Cref{def:semLA}), and the fact that $\actions \ne \varnothing$ and $E \ne \varnothing$.

   \item[\textbf{E2}] $e^{\alpha_i}_i \rightarrow K_i e^{\alpha_i}_i$: follows from \Cref{lemma:glob}.

   \item[\textbf{R1}] $\SP_i( U,\alpha_i|p) \leftrightarrow (pre(\alpha_i) \wedge p) \wedge e^{\alpha}_i$ :

    $\calM_i^v,v \models \langle  U,\alpha_i \rangle p \text{ and } \calM^{v,\alpha_i}_i \models e^{\alpha_i}_i $ (semantics of $\SP_i$)

    $\calM_i^v,v \models pre(\alpha_i) $ and $ \calM_i^{v,\alpha_i},(v,\alpha_i) \models p $ and $ \calM^{v,\alpha_i}_i \models e^{\alpha_i}_i$ (axioms of AML).

        \item[\textbf{R2}] $\SP_i( U,\alpha_i|\neg \varphi) \leftrightarrow (pre(\alpha_i) \wedge \neg \SP_i( U,\alpha_i|\varphi))$:

        $\calM_i^v,v \models \langle  U,\alpha_i \rangle \neg \varphi \text{ and } \calM^{v,\alpha_i}_i \models e^{\alpha_i}_i$ (sem. of $\SP_i$)

        $\calM_i^v,v \models (pre(\alpha_i) \wedge \neg \langle  U,\alpha_i \rangle \varphi) \text{ and } \calM^{v,\alpha_i}_i \models e^{\alpha_i}_i$ (axioms of AML)

        $\calM_i^v,v \models pre(\alpha_i) \wedge \neg \SP_i( U,\alpha_i|\varphi)$ (sem. of $\SP_i$, prop. reasoning)

        \item[\textbf{R3}] $\SP_i( U,\alpha_i|\varphi \wedge \psi) \leftrightarrow \SP_i( U,\alpha_i|\varphi) \wedge \SP_i( U,\alpha_i|\psi)$:

    $\calM_i^v,v \models \langle  U,\alpha_i \rangle \varphi \wedge \langle  U,\alpha_i \rangle \psi \text{ and } \calM^{v,\alpha_i}_i \models e^{\alpha_i}_i $ (sem. of $\SP_i$)

    $ (\calM_i^v,v \models\langle  U,\alpha_i \rangle \varphi $ and $ \calM^{v,\alpha_i}_i, \models e^{\alpha}_i) \wedge (\calM_i^v,v \models\langle U, \alpha_i \rangle \psi $ and $ \calM^{v,\alpha_i}_i \models e^{\alpha_i}_i)$ (taut., prop. reasoning)

    $\calM_i^v \models \SP_i( U,\alpha_i|\varphi) \wedge \SP_i( U,\alpha_i|\psi)$ (sem. of $\SP_i$)

        \item[\textbf{R4}] $\SP_i( U,\alpha_i|K_i\varphi) \leftrightarrow K_i \SP_i( U,\alpha_i|\varphi)$:

        $\calM_i^v,v \models \langle  U,\alpha_i \rangle K_i\varphi $ and $ \calM^{v,\alpha_i}_i \models e^{\alpha_i}_i$ (sem. of $\SP_i$)

        $\calM_i^v,v \models K_i \langle  U,\alpha_i \rangle \varphi $ and $ \calM^{v,\alpha_i}_i \models K_i(e^{\alpha_i}_i) $ (axioms of AML, \Cref{lemma:glob})

        $\calM_i^v \models K_i \SP_i( U,\alpha_i|\varphi)$ (property of $K_i$, sem. of $\SP_i$). 

        \item[\textbf{R5}] $\SP_i( U,\alpha_i|\langle U',\beta_i\rangle\varphi) \leftrightarrow \langle U\circ U', \alpha_i;\beta_i\rangle\varphi \wedge e^{\alpha_i}_i$:
      
        $\calM_i^v,v \models \langle U,\alpha_i\rangle\langle U',\beta_i\rangle\varphi $ and $ \calM^{v,\alpha_i}_i \models e^{\alpha_i}_i$ (sem. of $\SP_i$)

        $\calM_i^v,v \models pre(\alpha_i) $ and $ \calM^{v,\alpha_i}_i, (v,\alpha_i) \models pre(\beta_i)$ and $\calM^{v,\alpha_i;\beta_i}_i, (v,\alpha_i;\beta_i) \models \varphi $ and $ \calM^{v,\alpha_i}_i \models e^{\alpha_i}_i$ (sem. of $\langle U,\alpha_i\rangle$) 
        
        $\calM_i^v,v \models \langle U\circ U',\alpha_i;\beta_i\rangle\varphi $ and $ \calM^{v,\alpha_i}_i \models e^{\alpha_i}_i$  (def. of $\langle U,\alpha_i \rangle$,  \Cref{def:compo}.)

         \item[\textbf{R6}] $\SP_i( U,\alpha_i|\SP_i( U',\beta_i|\varphi)) \leftrightarrow \SP_i(U\circ U',\alpha_i;\beta_i|\varphi) \wedge e^{\alpha_i}_i$:

         $\calM_i^v,v \models \langle U,\alpha_i\rangle \SP_i( U',\beta_i|\varphi)$ and $ \calM^{v,\alpha_i}_i \models e^{\alpha_i}_i$ (sem. of $\SP_i$)
         
         $\calM_i^v,v \models pre(\alpha_i) $ and $ \calM^{v,\alpha_i}_i \models \SP_i( U,\beta_i|\varphi) $ and $ \calM^{v,\alpha_i}_i \models e^{\alpha_i}_i$ (sem. of $\langle U,\alpha_i\rangle$)
        
         $\calM_i^v,v \models pre(\alpha_i)$ and $ \calM^{v,\alpha_i}_i, (v,\alpha_i) \models \langle U,\beta_i\rangle\varphi $ and $ \calM^{v,(\alpha_i;\beta_i)}_i \models e^{\alpha_i;\beta_i}_i $ and $\calM^{v,\alpha_i}_i \models e^{\alpha_i}_i$ (sem. of $\SP_i$)

          $\calM_i^v,v \models pre(\alpha_i) $ and $ \calM^{v,\alpha_i}_i, (v,\alpha_i) \models pre(\beta_i) $ and $ \calM^{v,(\alpha_i;\beta_i)}_i, (v,(\alpha_i;\beta_i)) \models \varphi $ and $ \calM^{v,(\alpha_i;\beta_i)}_i \models e^{\alpha_i;\beta_i}_i $ and $ \calM^{v,\alpha_i}_i \models e^{\alpha_i}_i$ (sem. of $\langle U,\alpha_i\rangle$)

          $\calM_i^v \models \SP_i(U\circ U',\alpha_i;\beta_i|\varphi) $ and $ \calM^{v,\alpha_i}_i \models e^{\alpha_i}_i$ (sem. of $\SP_i$,  ,\Cref{def:compo}.)

    \end{itemize}
\end{proof}

We can now move to the completeness proof, for which we introduce the following definitions:

\begin{definition}[Translation]
\label{def:trans}
The translation $t: \La^{DAML} \rightarrow \La$ is defined as follows:
\setlength{\columnsep}{2pt}
\begin{multicols}{2}
\begin{itemize}
    \item $t(p)=p$
    \item $t(e^{\alpha_i}_i) =e^{\alpha_i}_i$ 
    \item $t(\neg \varphi)=\neg t(\varphi)$
    \item $t(\varphi \wedge \psi) =t(\varphi)\wedge t(\psi)$
    \item $t(K_i \varphi)= K_i t(\varphi)$

    \item $t(\langle U,\alpha_i\rangle p)$=
    \\
    $t( pre(\alpha_i)) \wedge p$

   \item $t(\langle U,\alpha_i\rangle\neg \varphi)$=
   \\
   $t(pre(\alpha_i) \wedge \neg \langle U,\alpha_i\rangle \varphi)$

   \item $t(\langle U,\alpha_i\rangle (\varphi \wedge \psi))$=
   \\
   $t(\langle U,\alpha_i\rangle\varphi \wedge \langle U,\alpha_i\rangle\psi)$

   \item $t(\langle U,\alpha_i\rangle K_i \varphi)$=
   \\
   $t(pre(\alpha_i) \wedge K_i \varphi)$

   \item $t(\langle U,\alpha_i\rangle \langle U',\beta_i\rangle \varphi)$= 
   \\
   $t(\langle U\circ U',\alpha_i;\beta_i\rangle \varphi)$

   \item $t (\SP_i( U,\alpha_i|p))$=
   \\
   $t(pre(\alpha_i) \wedge p \wedge e^{\alpha_i}_i)$
   
   \item $t(\SP_i( U,\alpha_i|\neg \varphi))$=
   \\
   $t(pre(\alpha_i) \wedge \neg \SP_i( U,\alpha_i|\varphi))$

   \item $t(\SP_i( U,\alpha_i|\varphi \wedge \psi))$= 
   \\
   $t(\SP_i( U,\alpha_i|\varphi) \wedge \SP_i( U,\alpha_i|\psi))$

   \item $t(\SP_i(\alpha_i|K_i\varphi))$= $t(K_i \SP_i( U,\alpha_i|\varphi))$

   \item $t(\SP_i( U,\alpha_i|\langle U',\beta_i\rangle\varphi))$= $t(\langle U\circ U',\alpha_i;\beta_i\rangle\varphi \wedge e^{\alpha_i}_i)$

   \item $t(\SP_i( U,\alpha_i|\SP_i( U',\beta_i|\varphi)))$= $t(\SP_i(U\circ U',\alpha_i;\beta_i|\varphi) \wedge e^{\alpha_i}_i)$
 
\end{itemize}
 \end{multicols}
    
\end{definition}

\begin{definition}[Complexity measure]
    We define a complexity function $c: \mathcal{L}^{DAML} \to \mathbb{N}$ by recursion on $\varphi$:

\begin{itemize}
    \item $c(p)=1$
    \item $c(e^{\alpha_i}_i)=1$
    \item $c(\neg \varphi)= 1 + c(\varphi)$
    \item $c(\varphi \wedge \psi)= 1+ Max(c(\varphi),c(\psi))$
    \item $c(K_{i} \varphi)=1 + c(\varphi)$
    \item $c(\langle U,\alpha_i\rangle\varphi)= (4+c(pre(\alpha_i)) \cdot c(\varphi)$
    \item $c (\SP_i( U,\alpha_i|\varphi))= (5+c(pre(\alpha_i)) \cdot c(\varphi) $
\end{itemize}
\end{definition}

It is easy to check that the complexity of any formula is strictly greater than the complexity of any proper sub-formula.

\begin{restatable}[]{theorem}{T}
\label{translation}
For every formula $\varphi \in \mathcal{L}^{DAML}$, $\vdash \varphi \leftrightarrow t(\varphi)$

\end{restatable}

\begin{proof}

We first prove the following auxiliary lemma  , showing that  the complexity of the formula on the left is greater than the complexity of the reduced formula on the right:
\begin{lemma}
\label{lem:>}
The following inequalities hold:
\begin{enumerate}

   \item $c (\SP_i( U,\alpha_i|p)) > c(pre(\alpha_i) \wedge p \wedge e^{\alpha_i}_i)$
   
   \item $c(\SP_i( U,\alpha_i|\neg \varphi)) > c(pre(\alpha_i) \wedge \neg \SP_i( U,\alpha_i|\varphi))$

   \item $c(\SP_i( U,\alpha_i|\varphi \wedge \psi)) > c(\SP_i( U,\alpha_i|\varphi) \wedge \SP_i( U,\alpha_i|\psi))$

   \item $c(\SP_i( U,\alpha_i|K_i\varphi)) > c(K_i \SP_i( U,\alpha_i|\varphi))$

   \item $c(\SP_i( U,\alpha_i|\langle U',\beta_i\rangle\varphi)) > c(\langle U\circ U',\alpha_i;\beta_i\rangle\varphi \wedge e^{\alpha_i}_i)$

   \item $c(\SP_i( U,\alpha_i|\SP_i( U',\beta_i|\varphi))) > c(\SP_i(U\circ U',\alpha_i;\beta_i|\varphi) \wedge e^{\alpha_i}_i)$ 
\end{enumerate}
 
\end{lemma}

\begin{proof}

\begin{enumerate}

    \item $c (\SP_i( U,\alpha_i|p)) > c(pre(\alpha_i) \wedge p \wedge e^{\alpha_i}_i)$
    
    $(5 + c(pre(\alpha_i))) \cdot c(p) > 1 + Max(pre(\alpha_i),c(p),c(e^{\alpha_i}_i))$

    \item $c(\SP_i( U,\alpha_i|\neg \varphi)) > c(pre(\alpha_i) \wedge \neg \SP_i( U,\alpha_i|\varphi))$

    $(5+ c(pre(\alpha_i))) \cdot (1 + c(\varphi)) > 1 + Max(c(pre(\alpha_i)), (1+ (5+ c(pre(\alpha_i)) \cdot c(\varphi))))$ 

    $(5+ c(pre(\alpha_i))) \cdot (1 + c(\varphi)) > 1 + Max(c(pre(\alpha_i)), 6+ c(pre(\alpha_i))\cdot c(\varphi))$

    \item $c(\SP_i( U,\alpha_i|\varphi \wedge \psi)) > c(\SP_i( U,\alpha_i|\varphi) \wedge \SP_i( U,\alpha_i|\psi))$

    $(5 + c(pre(\alpha_i))) \cdot (1 +Max( c(\psi), c(\varphi))) > 1 + Max((5+  c(pre(\alpha_i))) \cdot c(\varphi) , (5+  c(pre(\alpha_i))) \cdot c(\psi))$

   \item $c(\SP_i( U,\alpha_i|K_i\varphi)) > c(K_i \SP_i( U,\alpha_i|\varphi))$

   $(5 + c(pre(\alpha_i))) \cdot (c(\varphi)+1) > 1 + ((5 + c(pre(\alpha_i))) \cdot c(\varphi))$

   \item $c(\SP_i( U,\alpha_i|\langle U',\beta_i\rangle\varphi)) > c(\langle U\circ U',\alpha_i;\beta_i\rangle\varphi \wedge e^{\alpha_i}_i)$

   $(5 + (c(pre(\alpha_i))) \cdot ((4+ c(pre(\beta_i))) \cdot c(\varphi))) > 1+ Max ((4+ c(pre(\alpha_i;\beta_i)) \cdot c(\varphi)), c(e^{\alpha_i}_i))$

   $(5 + (c(pre(\alpha_i))) \cdot ((4+ c(pre(\beta_i))) \cdot c(\varphi))) > 1+ Max ((4+ 1 + Max(c(pre(\alpha_i),c(pre(\beta_i)))) \cdot c(\varphi)), c(e^{\alpha_i}_i))$

   \item $c(\SP_i( U,\alpha_i|\SP_i( U',\beta_i|\varphi))) > c(\SP_i(U\circ U',\alpha_i;\beta_i|\varphi) \wedge e^{\alpha_i}_i)$

   $(5 + (c(pre(\alpha_i)))) \cdot ((5+ c(pre(\beta_i))) \cdot c(\varphi)) > 1+ Max ((5+ 1 +Max(c(pre(\alpha_i),c(pre(\beta_i)) ))) \cdot c(\varphi), c(e^{\alpha_i}_i))$

\end{enumerate}

\end{proof}
    
We can now move to the proof of \Cref{translation}, which is by induction on $c(\varphi)$. All cases that do not involve the new modality are standard \cite{ditmarsch2007dynamic}. Here we show only the new cases:
\begin{itemize}
    \item For $\SP_i( U,\alpha_i|p)$: It follows from axiom \textbf{R1} of \Cref{def:redAx}, item 1 of \Cref{lem:>}, and the induction hypothesis.
    \item For $\SP_i( U,\alpha_i|\neg \varphi)$: It follows from axiom \textbf{R2} of \Cref{def:redAx}, item 2 of \Cref{lem:>}, and the induction hypothesis.
    \item For $\SP_i( U,\alpha_i|\varphi \wedge \psi)$: It follows from axiom \textbf{R3} of \Cref{def:redAx}, item 3 of \Cref{lem:>}, and the induction hypothesis.
    \item For $\SP_i( U,\alpha_i|K_i\varphi)$: It follows from axiom \textbf{R4} of \Cref{def:redAx}, item 4 of \Cref{lem:>}, and the induction hypothesis.
    \item For $\SP_i( U,\alpha_i|\langle U',\beta_i\rangle\varphi)$: It follows from axiom \textbf{R5} of \Cref{def:redAx}, item 5 of \Cref{lem:>}, and the induction hypothesis.
    
    \item For $\SP_i( U,\alpha_i|\SP_i( U',\beta_i|\varphi))$: It follows from axiom \textbf{R6} of \Cref{def:redAx}, item 6 of \Cref{lem:>}, and the induction hypothesis.
\end{itemize}

\end{proof}

\begin{theorem}[Completeness] For every formula $\varphi \in \mathcal{L}^{DAML}$

$$\models \varphi \text{ implies } \vdash \varphi$$
\end{theorem}

\begin{proof}
Suppose $\models \varphi$. Thus, 
$\models t (\varphi) $ by the soundness of the proof system (\Cref{thm:soundness}) and $\textbf{DAML} \vdash \varphi \leftrightarrow t(\varphi)$ (\Cref{translation}). In particular $t (\varphi)$ does not contain any deontic operator. Thus, $AML \vdash t (\varphi)$, since 
the logic for $\mathcal{L}$ is complete \cite{ditmarsch2007dynamic}. Consequently, $\textbf{DAML} \vdash t(\varphi)$.
Since $\textbf{DAML} \vdash \varphi \leftrightarrow t (\varphi)$, $ \textbf{DAML} \vdash \varphi$.

\end{proof}

In the next section, we illustrate the scope of our framework by modeling two examples, namely the miner's puzzle and a multi-agent example.

\section{Applications}
\label{sec:examples}

In this section, we illustrate our framework's ability in modeling the single-agent miner's puzzle, and a more convoluted two-agents example with asymmetric information.
Before moving to the examples, we make some important remarks:

\begin{remark}[Agents' perspective]
    \textbf{DAML} represents agents' hypothetical reasoning about the consequences of their available actions and in evaluating which choices yield the most desirable overall outcomes. Accordingly, the Kripke models presented in this section are interpreted from the perspective of an agent simulating and comparing action outcomes before committing to a decision.
\end{remark}

\begin{remark}[Actual world and actual action]
\label{rem:actual}
In our examples, the evaluation point of a Kripke model does not necessarily represent the actual world, as which world is actual depends on the action performed by the agent. In our framework, agents are not running experiments to discover which world is the actual world, as in standard epistemic logic, but rather they are simulating the potential effects of actions, in order to determine which one leads to the most desirable consequences.

Likewise, action models do not have an actual action being performed. They are all equally actual as part of a simulation. 
In addition, the precondition function of each action does not actually represent only the precondition for performing that action - rather it also represents its epistemic effects on the model. We kept the name \textit{precondition function} in line with the existing literature.
\end{remark}

\subsection{Miner's puzzle}
\label{sec:miner}

We model the miner's puzzle from the perspective of a single agent $i$ deciding which shaft (if any) to block.
Since in the single-agent case $\calM = \calM_i$, we will omit the agent subscript on models.

The initial situation is depicted in \Cref{fig:Miners_AM} (top left). 
Each world corresponds to a possible outcome, with one variable indicating the position of the miners (shaft $A$ or shaft $B$) and the other representing the number of lives that can be saved, which in this scenario also corresponds to the desirability value of each world: $f(A \wedge 10)= 10, f(A \wedge 0)= 0$ and $f(A \wedge 9)= 9$ and analogously for the worlds where $B$ holds.
We assume complete ignorance of agent $i$ concerning the current state of affairs.

Agent $i$ can perform three actions, as represented in the action model $U$ of \Cref{fig:Miners_AM}(top right):
Action $\alpha_i$ (block shaft $A$) can save ten lives (if the miners are in the shaft $A$), or zero otherwise: $pre(\alpha_i)= (A \wedge 10) \vee (B \wedge 0)$. 
Analogously for action $\beta_i$ (block shaft $B$), $pre(\beta_i)= (A \wedge 0) \vee (B \wedge 10)$.
Blocking neither shaft (action $\gamma_i$) will ensure the saving of nine lives, independently of the location of the miners: $pre(\gamma_i)=9$.
None of the actions informs $i$ about the actual location of the miners.

$\calM \otimes U$ \Cref{fig:Miners_AM} (below) is the result of $i$'s reasoning about the epistemic consequences of its available actions, and it consists of three disconnected action-generated submodels: $\calM^{\alpha_i}, \calM^{\beta_i}, \calM^{\gamma_i}$.
Based on this results, we can evaluate deontic statements in the initial model. For example, we can check whether $i$ ought to block shaft $A$ given the current information, which amounts to check formula $\calM \models \SP_i ( U,\alpha_i|\top)$. By \Cref{def:semLA} it amounts to check the truth of the formula: $\calM,v \models \langle  U,\alpha_i \rangle \top \text{ and } \calM^\alpha \models e^{\alpha_i}_i$.
While the first conjunct is true, the second one is not: $\calM^\alpha \not \models e^{\alpha_i}_i$, as $\E_i(\calM^\gamma) > \E_i(\calM^\alpha) = \E_i(\calM^\beta)$. Thus, $\calM \not \models \SP_i ( U,\alpha_i|\top)$.
Similarly, $\calM \not\models \SP_i ( U,\beta_i|\top)$. On the other hand, $\calM \models \SP_i ( U,\gamma_i|\top)$, meaning that for agent $i$, given the current information, blocking neither shaft leads to the highest expected deontic value.

\begin{figure}[t]
    \centering
    \resizebox{9cm}{!}{%

\begin{tikzpicture}
    \begin{scope}[]
    \node[circle,draw,minimum size=12mm, label=135:$10$] (1) {$A,10$};
    \node (g') [above = .5cm  of 1]  {\Large{$\calM$}};
    \node (g) [right = .9cm  of 1]  {};
    \node[circle,draw,minimum size=12mm, label=135:$9$,label=45:$v$] (2) [above = .5cm  of g]  {$\underline{A,9}$};
    \node[circle,draw,minimum size=12mm, label=45:$0$] (3) [right = 2cm  of 1] {$A,0$}; 
    \node[circle,draw,minimum size=12mm,, label=135:$10$] (4) [below = 1cm  of 1] {$B,10$};
    \node (h) [right = .9cm  of 4]  {};
    \node[circle,draw,minimum size=12mm,, label=180:$9$] (5) [below = .5cm  of h]  {$B,9$};
    \node[circle,draw,minimum size=12mm, label=45:$0$] (6) [below = 1cm  of 3] {$B,0$};

    \node[label=135:$\alpha_i$, draw,minimum size=10mm,text width=1.5cm] (7) [right = 2cm  of 5] {$(A \wedge 10) \vee$ \\ $(B \wedge 0)$};
    \node[draw, label=45:$\beta_i$,minimum size=10mm, text width=1.5cm ] (8) [right = 1cm  of 7]  {$(A \wedge 0) \vee$  \\ $(B \wedge 10)$};
    \node (a) [right = .5cm  of 7] {};
    \node[draw, label=135:$\gamma_i$, minimum size=10mm] (9) [above = 1cm  of a]  {$9$};
    \node (b) [left = 1cm  of 9] {\Large{$U$}};

\node[circle,draw,minimum size=12mm, label=135:\large{$\calM^{\alpha_i}$}, label=180:$10$] (10) [below = 3.5cm  of 5] {$A,10$};
    
    \node (c) [right = .25cm  of 10]  {};
    
    \node[] (11) [below = .25cm  of 10] {};
    
    \node[circle,draw,minimum size=12mm,label=180:$0$] (12) [right = .5cm  of 11] {$B,0$};  
    
    \node[circle,draw,minimum size=12mm,, label=45:\large{$\calM^{\beta_i}$}, label=0:$0$] (13) [right = 2.75cm  of 10] {$A,0$};
    
    \node[circle,draw,minimum size=12mm,label=0:$10$] (14) [right = .5cm  of 12] {$B,10$};
    
    \node (c') [right = 1cm  of 10]  {};
    
    \node[circle,draw,minimum size=12mm,label=180:$9$] (15) [above = .5cm  of c'] {$B,9$};
    
    \node[circle,draw,minimum size=12mm, label=135:\large{$\calM^{\gamma_i}$},label=180:$9$,label=45:$v$] (16) [above = .5cm  of 15] {$\underline{A,9}$};
    
    \node (h) [left = 1cm  of 16]  {\Large{$\calM \otimes U$}};

    \path[]
        (1) edge node[above] {} (2)
        (1) edge node[above] {} (3)
        (3) edge node[above] {} (2)
        (5) edge node[above] {} (2)
        (1) edge node[above] {} (4)
        (1) edge node[above] {} (6)%
        (3) edge node[above] {} (6)
        (3) edge node[above] {} (4)%
        (4) edge node[above] {} (6)
        (4) edge node[above] {} (5)
        (5) edge node[above] {} (6);

     \path[->] 
          (7) edge[loop left,looseness=5] node[left] {} (7)
          (8) edge[loop right,looseness=5] node[right] {} (8)
          (9) edge[loop right,looseness=5] node[right] {} (9);

    \path[]
        (10) edge node[above] {} (12)
        (13) edge node[above] {} (14)
        (16) edge node[left] {} (15)
        ;
         
    \end{scope}
\end{tikzpicture}
    }
    \caption{Initial pointed  graded Kripke model $\calM,v$ of the miner's puzzle (left). The evaluation point $v$ is underlined. 
    Reflexive loops and agent-labels on relations are omitted. Numbers close to the worlds represents their desirability value.
    Action model $U$ representing $i$'s possible actions (right). 
    Updated model $\calM \otimes U$ (below) made of three action-generated submodels, one for each action, $\calM^{\alpha_i}$, $\calM^{\beta_i}$, $\calM^{\gamma_i}$.
    Reflexive loops are omitted.}
    \label{fig:Miners_AM}
\end{figure}
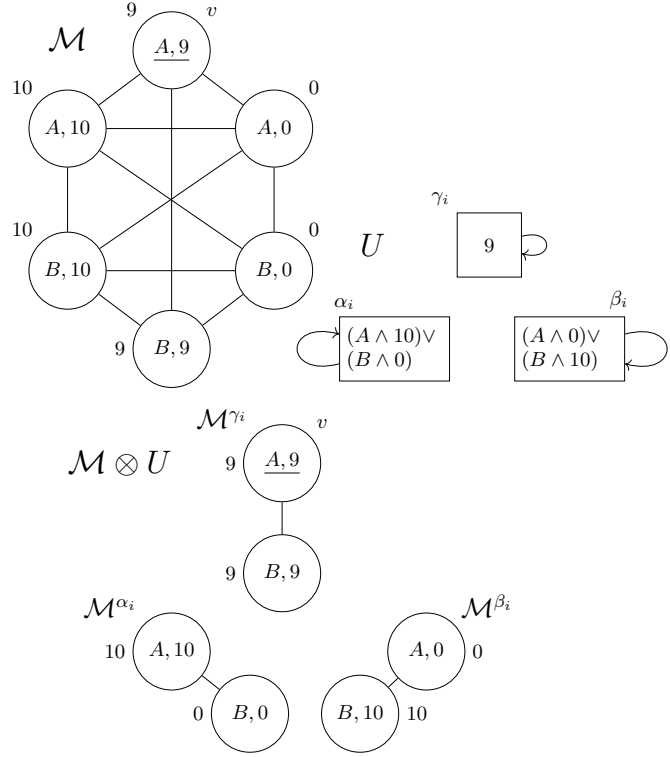

\subsection{Obligation to inform}
\label{sec:med}

A patient $p$ has condition $C$ and is treated at the hospital where Alice ($a$) and Bethany ($b$) work. 
The condition can be treated with two drugs, $d$ or $d'$. Drug $d$ has almost a $100\%$ success rate for condition $C$, unless an allergy is present, which makes it completely ineffective. Drug $d'$, has a $40\%$ success rate both for people who are allergic to $d$, and for those who aren't. It is up to $a$ to administer the drug.
$p$ is allergic to drug $d$, and $a$ does not know that.
On the other hand, $b$ knows that $p$ is allergic to the drug, and knows that $a$ believes that $b$ also does not know.

We model the example using two agents, $a$ and $b$, and we assume the viewpoint of the latter.
The initial model is represented in \Cref{fig:Allergy1} (left). Because of the asymmetry in agent's knowledge, the model is \textbf{KD45}. Each world has two variables, one indicating the presence (resp. absence) of the allergy $A$ (resp. $\neg A$) and one for which drugs is administered ($d$ or $d'$). The desirability values correspond to the success rate of the treatment, abstracting away from other details.
In particular, The top two worlds are accessible only to $b$: she knows $p$ has allergy $A$. She also knows that $a$ does not know, and that $a$ believes that also $b$ does not know, as represented by the four remaining worlds.
The purpose of the modeling is to represent $b$'s hypothetical reasoning, simulating the effect of the combined actions available to both agents.

We use two action models, $U$ and $U'$, one decision point per agent. $U$ \Cref{fig:Allergy1} (top left) represents $b$'s choice between informing $a$ of the allergy (action $\delta_b$ with $pre(\delta_b)= A$) or simply doing nothing (action $\gamma_b$ with  $pre(\gamma)=\top$). $U'$ \Cref{fig:Allergy1} (bottom left) represents $a$'s choice between administering the drug $d$ (action $\alpha_i$ with $pre(\alpha_a)= d$) or $d'$ (action $\beta_i$ with $pre(\beta_a)=d'$).

$\calM \otimes U$ (\Cref{fig:Allergy2}) represents the possible outcomes of $b$'s actions in its two action-generated submodels $\calM^{v,\delta_b}$ and $\calM^{v,\gamma_b}$.

$\calM \otimes U \otimes U'$ (\Cref{fig:Allergy3}) represents the result of $a$'s possible actions executed after $b$'s possible actions, generating four action-generated submodels, from left to right $\calM^{w,(\delta_b;\alpha_a)}$, $\calM^{v,(\delta_b;\beta_a)}$, $\calM^{w,(\gamma_b;\alpha_a)}$ and $\calM^{v,(\gamma_b;\beta_a)}$.
Each action-generated submodel represents the consequences of each combination of actions for the two agents\footnote{This model can be thought of a game-tree in game theoretic settings. Exploring the connections to game theory is left for future work.}.

Based on these resulting submodels, it is possible compute agent's obligations in the initial model. 
In particular, $a$'s expected deontic value for each action generated submodel is: $\E_a(\calM^{w,(\delta_b;\alpha_a)})= 0$, $\E_a(\calM^{v,(\delta_b;\beta_a)}) = \E_a(\calM^{v,(\gamma_b;\beta_a)}) $$= 40$, $\E_a(\calM^{w,(\gamma_b;\alpha_a)})= 50$. This means that the best possible outcome is where $p$ is not allergic and drug $d$ is administered. But this case is considered possible only by $a$ in case she is not informed that $p$ is allergic. From $b$'s perspective, $\E_b(\calM^{w,(\delta_b;\alpha_a)})= 0$, $\E_b(\calM^{v,(\delta_b;\beta_a)}) = \E_b(\calM^{v,(\gamma_b;\beta_a)})$ $= 40$, $\E_b(\calM^{w,(\gamma_b;\alpha_a)})= 0$, because she knows $p$ is allergic. In particular, the expected deontic value of $b$ is maximal whenever $a$ administers $d'$.

Also note that $\calM_b^v \models K_b \SP_a ( U',\beta_a|A)$, i.e., $b$ knows that $a$ should perform $\beta_a$ if $A$ were the case, and $\calM_b^v \models K_b \SP_a ( U',\alpha_a|\top)$, i.e., $b$ knows that, without further information, $a$ should instead perform $\alpha_a$.

The goal of this modeling is to capture the fact that, if $a$ knew that $A$ (as $b$ already does), $a$ should perform $\beta_a$, reason why $b$ should perform $\delta_b$.
This statement can be expressed using the formula: $\calM_b^v \models \SP_b ( U,\delta_b | \SP_a ( U',\beta_a |K_aA))$. By semantic reasoning (\Cref{def:semLA}), we can unfold it in to the equivalent:

$\calM_b^v,v \models pre(\delta_b)$ and $\calM^{v,\delta_b}_b, (v,\delta_b) \models pre(\beta_a)$ and $\calM^{v,(\delta_b;\beta_a)}_a, (v,(\delta_b;\beta_a)) \models K_aA$ and $\calM^{v,(\delta_b;\beta_a)}_a \models e^{\delta_b;\beta_a}_a$ and $\calM^{v,\delta_b}_b \models e^{\delta_b}_b$.

The first two conjuncts are easily checked, as $pre(\delta_b)= A$ and $pre(\beta_a) = d'$, which respectively hold in $\calM_b^v,v$ of \Cref{fig:Allergy1} (left),
and in
$\calM^{v,\delta_b}_a, (v,\delta_b)$ of \Cref{fig:Allergy2} (left).
The third conjunct is evaluated in the model $\calM^{v,(\delta_b;\beta_a)}_a$ of \Cref{fig:Allergy3} (second from the left)), resulting from applying $\delta_b$ followed by $\beta_a$, where $K_a A$ holds as in every world that is $a$-accessible, $A$ holds.
Finally, the last two conjuncts hold as $\E_b(\calM^\delta_b) \geq \E_b(\calM^\gamma_b)$ and $\E_a(\calM^{\delta;\beta}_a) \geq \E_a(\calM^{\delta;\alpha}_a)$.

Similarly, it can be checked that $\calM_b^v \not\models \SP_b ( U,\gamma_b | \SP_a ( U',\beta_a |K_aA))$, which means that $b$ should not avoid informing $a$, as $a$ would perform $\beta_a$ if $a$ knew $A$ (as $b$ does). 
What makes the formula fail is the fact that $\calM^{v,\gamma_b}_a, (v,\gamma_b) \models \langle  U',\beta_a\rangle K_aA$ does not hold, as $\calM^{v,(\gamma_b;\beta_a)}_a, (v,(\gamma_b;\beta_a)) \not\models K_aA$, since action $\gamma_b$ does not enforce $A$ in $a$'s submodel.

\begin{figure}
 \centering
  \resizebox{9.2cm}{!}{%
\begin{tikzpicture}
\node[circle, draw, minimum size=14mm, label=135:$0$,label=90:$w$] (1) {$A, d$};
    \node[circle, draw, minimum size=14mm,label=45:$40$,label=90:$v$] (2) [right = 1cm  of 1]  {$\underline{A, d'} $};
    \node[circle, draw, minimum size=14mm,label=135:$0$] (3) [below = 1cm  of 1]  {$A, d$};
    \node[circle, draw, minimum size=14mm,label=45:$40$] (4) [right = 1cm  of 3]  {$ A, d'$};
    \node[circle, draw, minimum size=14mm,label=135:$100$] (5) [below = 1cm  of 3]  {$ \neg A, d$};
    \node[circle, draw, minimum size=14mm,label=45:$40$] (6) [right = 1cm  of 5] {$\neg A, d'$};
    \node (h) [right = .5cm  of 1] {};
    \node (q) [above = 1cm  of h] {\Large{$\calM$}};
    
    \node[] (W) [right = 2.5cm  of 2] {};
    \node[ label=135:$\delta_b$, draw,minimum size=8mm] (7) [below = .5cm  of W] {$A$};
    \node[ draw, minimum size=8mm, label=45:$\gamma_b$] (8) [right = 1cm  of 7]  {$\top$};
    \node (ha) [below = .5cm  of 7] {};
    \node (h) [right = .6cm  of ha] {\Large{$U$}};

    \path[->] 
          (7) edge[loop left,looseness=4,in=205,out=240] node[left] {$a,b$} (7)
          (8) edge[loop right,looseness=4,in=295,out=330] node[right] {$a,b$} (8);

    \node[label=135:$\alpha_a$, draw,minimum size=8mm] (9) [below = 2cm  of 7] {$d$};
    \node[draw, minimum size=8mm, label=45:$\beta_a$] (10) [right = 1cm  of 9]  {$d'$};
    \node (da) [below = .5cm  of 9] {};
    \node (d) [right = .6cm  of da] {\Large{$U'$}};

    \path[->] 
          (9) edge[loop left,looseness=4,in=205,out=240] node[left] {$a,b$} (9)
          (10) edge[loop right,looseness=4,in=295,out=330] node[right] {$a,b$} (10);

    \path[]
          (1) edge[] node[above] {$b$} (2)
          (3) edge[] node[below] {$a,b$} (4)
          (3) edge[] node[right] {$a,b$} (5)
          (6) edge[] node[above] {$a,b$} (5)
          (4) edge[] node[left] {$a,b$} (6)
          ;
          
    \path[-{Stealth[length=2mm, width=2mm]}] 
          (1) edge[] node[left] {$a$} (3)
          (1) edge[] node[left] {$a$} (4)
          (2) edge[] node[below] {$a$} (3)
          (2) edge[] node[left] {$a$} (4);

    \path[->]
          (1) edge [loop left,looseness=4,in=205,out=240] node[left] {$b$} (1)
          (2) edge [loop right,looseness=4,in=295,out=330] node[right] {$b$} (2)
          (3) edge [loop left,looseness=4,in=205,out=240] node[left] {$a,b$} (3)
          (4) edge [loop right,looseness=4,in=295,out=330] node[right] {$a,b$} (4)
          (5) edge [loop left,looseness=4,in=205,out=240] node[left] {$a,b$} (5)
          (6) edge [loop right,looseness=4,in=295,out=330] node[right] {$a,b$} (6);

\begin{pgfonlayer}{background}
\filldraw [line width=6mm,line join=round,black!10]
      (1.north  -| 2.east)  rectangle (1.south  -| 1.west);

\filldraw [line width=6mm,line join=round,black!10]
      (3.north  -| 4.east)  rectangle (6.south  -| 3.west);
\end{pgfonlayer}
       
\end{tikzpicture}

    }
    \caption{Initial  graded Kripke model $\calM$ (left).
    Transitive arrows are omitted. The evaluation point $v$ is underlined. Gray rectangles represent agents' submodels from $v$, the top being $\calM^v_b$ and the bottom is $\calM^v_a$. 
    Action model $U$ (right top), representing $b$'s possible actions $\delta_b$ and $\gamma_b$. Action model $U'$ (right bottom), representing $a$'s possible actions $\alpha_a$ and $\beta_a$.}
    \label{fig:Allergy1}
\end{figure}
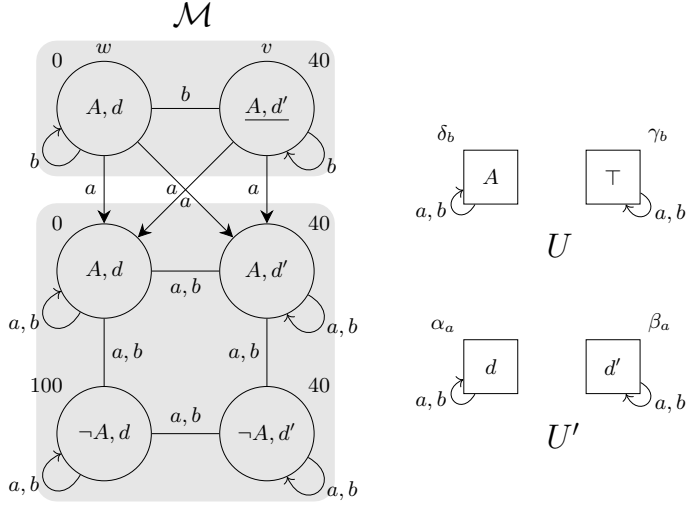

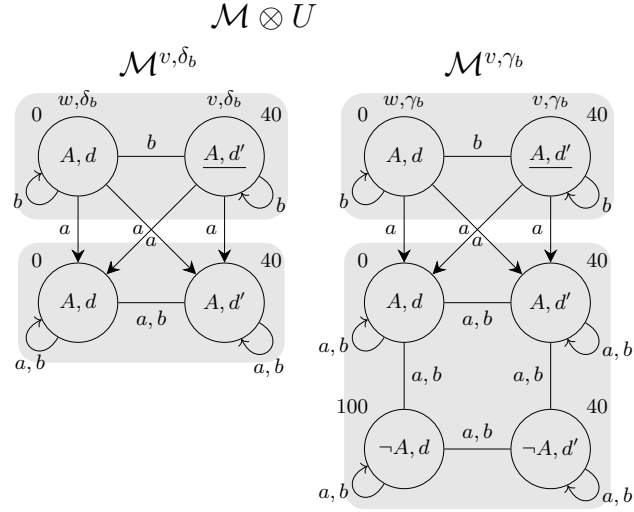
\begin{figure}
 \centering
\resizebox{8.5cm}{!}{
\begin{tikzpicture}
    \node[circle, draw, minimum size=12mm, label=135:$0$,label=90:$w{,}\delta_b$] (1) {$A, d$};
    \node[circle, draw, minimum size=12mm,label=45:$40$, label=90:$v{,}\delta_b$] (2) [right = 1cm  of 1]  {$\underline{A, d'} $};
    \node[circle, draw, minimum size=12mm,label=135:$0$] (3) [below = 1cm  of 1]  {$A, d$};
    \node[circle, draw, minimum size=12mm,label=45:$40$] (4) [right = 1cm  of 3]  {$ A, d'$};
    \node (g) [right = .5cm  of 1]  {};
    \node (g) [above = 1cm  of g]  {\Large{$\calM^{v,\delta_b}$}};

    \path[]
          (1) edge[] node[above] {$b$} (2)
          (3) edge[] node[below] {$a,b$} (4);
          
   \path[-{Stealth[length=2mm, width=2mm]}] 
          (1) edge[] node[left] {$a$} (3)
          (1) edge[] node[left] {$a$} (4)
          (2) edge[] node[below] {$a$} (3)
          (2) edge[] node[left] {$a$} (4);

    \path[->]
          (1) edge [loop left,looseness=4,in=205,out=240] node[left] {$b$} (1)
          (2) edge [loop right,looseness=4,in=295,out=330] node[right] {$b$} (2)
          (3) edge [loop left,looseness=4,in=205,out=240] node[below] {$a,b$} (3)
          (4) edge [loop right,looseness=4,in=295,out=330] node[below] {$a,b$} (4);

  \begin{pgfonlayer}{background}
\filldraw [line width=7mm,line join=round,black!10]
      (1.north  -| 2.east)  rectangle (1.south  -| 1.west);

\filldraw [line width=6mm,line join=round,black!10]
      (3.north  -| 4.east)  rectangle (3.south  -| 3.west);
\end{pgfonlayer}

\node[circle, draw, minimum size=12mm, label=135:$0$,label=90:$w{,}\gamma_b$] (7) [right = 1.5cm  of 2] {$A, d$};
    \node[circle, draw, minimum size=12mm,label=45:$40$,label=90:$v{,}\gamma_b$] (8) [right = 1cm  of 7]  {$\underline{A, d'} $};
    \node[circle, draw, minimum size=12mm,label=135:$0$] (9) [below = 1cm  of 7]  {$A, d$};
    \node[circle, draw, minimum size=12mm,label=45:$40$] (10) [right = 1cm  of 9]  {$ A, d'$};
    \node (c) [right = .5cm  of 7]  {};
    \node (c') [above = 1cm  of c]  {\Large{$\calM^{v,\gamma_b}$}};
 
    \node[circle, draw, minimum size=12mm,label=135:$100$] (11) [below = 1cm  of 9]  {$ \neg A, d$};
    \node[circle, draw, minimum size=12mm,label=45:$40$] (12) [right = 1cm  of 11] {$\neg A, d'$};

\node (h) [left = 2.5cm  of c'] {};
\node (w) [above = .3cm  of h] {\Large{$\calM \otimes U$}};

    \path[]
          (7) edge[] node[above] {$b$} (8)
          (9) edge[] node[below] {$a,b$} (10)
          (9) edge[] node[right] {$a,b$} (11)
          (12) edge[] node[above] {$a,b$} (11)
          (10) edge[] node[left] {$a,b$} (12)
          ;
          
    \path[-{Stealth[length=2mm, width=2mm]}] 
          (7) edge[] node[left] {$a$} (9)
          (7) edge[] node[left] {$a$} (10)
          (8) edge[] node[below] {$a$} (9)
          (8) edge[] node[left] {$a$} (10);

    \path[->]
          (7) edge [loop left,looseness=4,in=205,out=240] node[left] {$b$} (7)
          (8) edge [loop right,looseness=4,in=295,out=330] node[right] {$b$} (8)
          (9) edge [loop left,looseness=4,in=205,out=240] node[left] {$a,b$} (9)
          (10) edge [loop right,looseness=4,in=295,out=330] node[right] {$a,b$} (10)
          (11) edge [loop left,looseness=4,in=205,out=240] node[left] {$a,b$} (11)
          (12) edge [loop right,looseness=4,in=295,out=330] node[right] {$a,b$} (12);
          
\begin{pgfonlayer}{background}
\filldraw [line width=7mm,line join=round,black!10]
      (7.north  -| 8.east)  rectangle (7.south  -| 7.west);

\filldraw [line width=6mm,line join=round,black!10]
      (9.north  -| 10.east)  rectangle (12.south  -| 9.west);
\end{pgfonlayer}

\end{tikzpicture}
    }
    \caption{Updated model $\calM \otimes U$, with corresponding action-generated submodels $\calM^{v,\delta_b}$ (left) and $\calM^{v,\gamma_b}$ (right), representing the epistemic outcomes of $b$'s available actions. Transitive arrows are omitted. The evaluation points $(v,\gamma)$ and $(v,\delta)$ are underlined. In each action-generated submodel, gray rectangles represent the agent-based submodels from $v,\delta$, and respectively $v,\gamma$, the top one for agent $b$ and the bottom for $a$ in both.
    }
    \label{fig:Allergy2}
\end{figure}

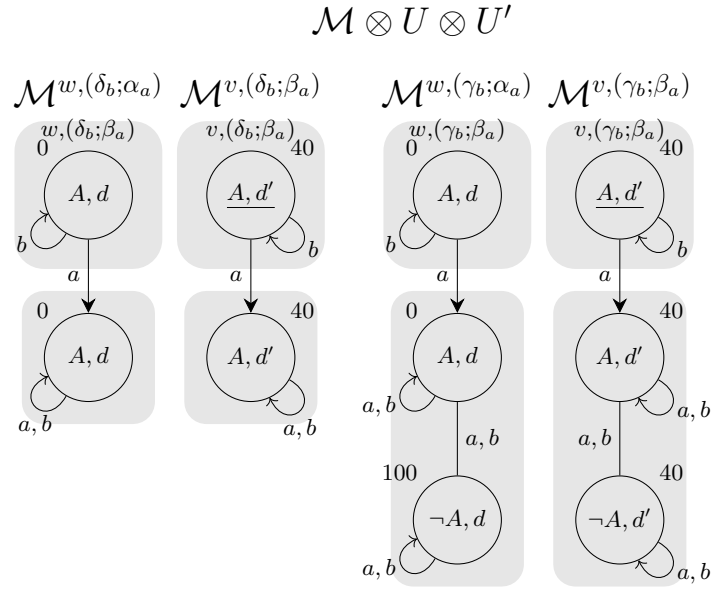
\begin{figure}
 \centering
\resizebox{9.5cm}{!}{%
\begin{tikzpicture}

    \node[circle, draw, minimum size=12mm, label=135:$0$,label=90:$w{,}(\delta_b{;}\beta_a)$] (1) {$A, d$};
    \node[circle, draw, minimum size=12mm,label=45:$40$,label=90:$v{,}(\delta_b{;}\beta_a)$] (2) [right = 1cm  of 1]  {$\underline{A, d'} $};
    \node[circle, draw, minimum size=12mm,label=135:$0$] (3) [below = 1cm  of 1]  {$A, d$};
    \node[circle, draw, minimum size=12mm,label=45:$40$] (4) [right = 1cm  of 3]  {$ A, d'$};
    \node (g) [above = .5cm  of 2]  {\Large{$\calM^{v,(\delta_b;\beta_a)}$}};
    \node (g') [above = .5cm  of 1]  {\Large{$\calM^{w,(\delta_b;\alpha_a)}$}};

    \node (h) [right = 1cm  of g] {};
    \node (h1) [above = .5cm  of h] {\Large{$\calM \otimes U \otimes U'$}};

   \path[-{Stealth[length=2mm, width=2mm]}] 
          (1) edge[] node[left] {$a$} (3)
          (2) edge[] node[left] {$a$} (4);

    \path[->]
          (1) edge [loop left, looseness=4,in=205,out=240] node[left] {$b$} (1)
          (2) edge [loop right, looseness=4,in=295,out=330] node[right] {$b$} (2)
          (3) edge [loop left, looseness=4,in=205,out=240] node[below] {$a,b$} (3)
          (4) edge [loop right, looseness=4,in=295,out=330] node[below] {$a,b$} (4);

  \begin{pgfonlayer}{background}
\filldraw [line width=8mm,line join=round,black!10]
      (1.north  -| 1.east)  rectangle (1.south  -| 1.west);

\filldraw [line width=8mm,line join=round,black!10]
      (2.north  -| 2.east)  rectangle (2.south  -| 2.west);

\filldraw [line width=6mm,line join=round,black!10]
      (3.north  -| 3.east)  rectangle (3.south  -| 3.west);

\filldraw [line width=6mm,line join=round,black!10]
      (4.north  -| 4.east)  rectangle (4.south  -| 4.west);
\end{pgfonlayer}


\node[circle, draw, minimum size=12mm, label=135:$0$,label=90:$w{,}(\gamma_b{;}\beta_a)$] (7) [right = 1.6cm  of 2] {$A, d$};
    \node[circle, draw, minimum size=12mm,label=45:$40$,label=90:$v{,}(\gamma_b{;}\beta_a)$] (8) [right = 1cm  of 7]  {$\underline{A, d'} $};
    \node[circle, draw, minimum size=12mm,label=135:$0$] (9) [below = 1cm  of 7]  {$A, d$};
    \node[circle, draw, minimum size=12mm,label=45:$40$] (10) [right = 1cm  of 9]  {$ A, d'$};
   
    \node (c) [above = .5cm  of 8]  {\Large{$\calM^{v,(\gamma_b;\beta_a)}$}};
    \node (c1) [above = .5cm  of 7]  {\Large{$\calM^{w,(\gamma_b;\alpha_a)}$}};
 
    \node[circle, draw, minimum size=12mm,label=135:$100$] (11) [below = 1cm  of 9]  {$ \neg A, d$};
    \node[circle, draw, minimum size=12mm,label=45:$40$] (12) [right = 1cm  of 11] {$\neg A, d'$};


    \path[]
          (9) edge[] node[right] {$a,b$} (11)
          (10) edge[] node[left] {$a,b$} (12)
          ;
          
    \path[-{Stealth[length=2mm, width=2mm]}] 
          (7) edge[] node[left] {$a$} (9)
          (8) edge[] node[left] {$a$} (10);

    \path[->]
          (7) edge [loop left, looseness=4,in=205,out=240] node[left] {$b$} (7)
          (8) edge [loop right, looseness=4,in=295,out=330] node[right] {$b$} (8)
          (9) edge [loop left, looseness=4,in=205,out=240] node[left] {$a,b$} (9)
          (10) edge [loop right, looseness=4,in=295,out=330] node[right] {$a,b$} (10)
          (11) edge [loop left, looseness=4,in=205,out=240] node[left] {$a,b$} (11)
          (12) edge [loop right, looseness=4,in=295,out=330] node[right] {$a,b$} (12);

\begin{pgfonlayer}{background}
\filldraw [line width=8mm,line join=round,black!10]
      (7.north  -| 7.east)  rectangle (7.south  -| 7.west);

\filldraw [line width=8mm,line join=round,black!10]
      (8.north  -| 8.east)  rectangle (8.south  -| 8.west);

\filldraw [line width=6mm,line join=round,black!10]
      (9.north  -| 9.east)  rectangle (11.south  -| 9.west);

\filldraw [line width=6mm,line join=round,black!10]
      (10.north  -| 10.east)  rectangle (12.south  -| 10.west);
\end{pgfonlayer}

\end{tikzpicture}

    }
    \caption{Updated model $\calM \otimes U \otimes U'$, representing all the epistemic outcomes of the actions of $b$ followed by the actions of $a$, with four corresponding action-generated submodels, from left to right: $\calM^{w,(\delta_b;\alpha_a)}, \calM^{v,(\delta_b;\beta_a)}, \calM^{w,(\gamma_b;\alpha_a)}$ and $\calM^{v,(\gamma_b;\beta)}$. 
    Transitive arrows are omitted. 
    In each action-generated submodel, gray rectangles represent the agent-based submodels the top one for $b$ and the bottom for $a$. 
    }
    \label{fig:Allergy3}
\end{figure}

\section{Conclusion and future work}
\label{sec:conclusion}

This paper introduced Deontic Action Model Logic (\textbf{DAML}), a novel multi-agent deontic modal logic that integrates the dynamic machinery of action model logic with a value-based evaluation of actions. 
From a knowledge representation perspective, \textbf{DAML} offers a structured way to encode and reason about normative information, epistemic uncertainty, and action effects within a unified logical framework. This makes it suitable for modeling and analyzing decision problems faced by autonomous agents operating in uncertain and normatively constrained domains.
We established soundness and completeness for the logic and illustrated its expressive power through applications to deontic scenarios.

There are several promising directions for future work: On the epistemic side, introducing explicit probability measures over worlds would yield a more fine-grained notion of expected deontic value, while replacing the knowledge operator with a graded belief operator would allow the framework to capture varying degrees of uncertainty about states of affairs and action outcomes. On the deontic side, alternative ways of computing expectation functions could support additional modalities, such as permissions and prohibitions. Finally, enriching the action language with more expressive constructs, including non-deterministic choice, would enable the study of conflicting obligations and more sophisticated forms of decision-making in dynamic multi-agent settings.

\section*{Funding}
This work was supported by the Austrian Science Fund (FWF) project A Logical Framework for Graded Deontic Reasoning [10.55776/PAT2141924] (\doi{10.55776/PAT2141924}).

\section*{Acknowledgments}

I am thankful to Christian Fermüller, Xavier Parent and Henri Thölke
for the multiple illuminating discussions.

\bibliographystyle{plainurl}

\bibliography{references}

\end{document}